\pdfoutput=1
\documentclass{llncs}
\usepackage{xspace}
\usepackage{xcolor}

\usepackage{amssymb}
\usepackage{amsmath}
\usepackage{booktabs}
\usepackage{footnote}
\usepackage{microtype}
\newcommand{\comprehension}[2]{\ensuremath{\left\{ {#1} \;|\; {#2}\right\}}}

\newcommand{\limp}{\to}

\newtheorem{observation}{Observation}
\usepackage{rotating}
\usepackage{subcaption}
\usepackage{tikz}

\def\miko{Mikol\'a\v{s} Janota}
\def\jpms{Joao Marques-Silva}
\def\rgrig{Radu Grigore}
\def\theTitle{On QBF Proofs and Preprocessing%
}
\title{\theTitle}
\author{{\miko} \inst{1}\and {\rgrig}\inst{3}\and {\jpms}\inst{1,2}}
\institute{INESC-ID, Lisbon, Portugal %
\and University College Dublin, Ireland %
\and University of Oxford, UK}
\definecolor{citeblue}{rgb}{0.1,0,.4}
\definecolor{refcolor}{rgb}{0,0,0.4}
\definecolor{midgreen}{RGB}{0,150,0}
\definecolor{darkgreen}{RGB}{0,128,0}
\definecolor{darkblue}{RGB}{0,0,128}
\definecolor{darkred}{RGB}{192,0,0}
\usepackage[%
colorlinks=true%
,bookmarks=true%
,linkcolor=citeblue%
,citecolor=citeblue%
,urlcolor=blue%
,plainpages=false]{hyperref}

\def\satelite{\xspace{\sf SatELite}\xspace}

\DeclareMathOperator*{\vars}{{\sf var}}
\DeclareMathOperator*{\union}{\cup}

\DeclareMathOperator*{\Pref}{{\mathcal P}}
\DeclareMathOperator*{\Cert}{{\mathcal C}}

\hbadness=\pretolerance

\begin{document}
\maketitle
\setcounter{footnote}0 
\begin{abstract}
QBFs (quantified boolean formulas), which are a superset of propositional
formulas, provide a canonical representation for PSPACE problems.
To overcome the inherent complexity of QBF, significant
effort has been invested in developing QBF solvers as well as the
underlying proof systems.
At the same time, formula preprocessing is crucial for the application
of QBF solvers. This paper focuses on a missing link in
currently-available technology: How to obtain a certificate
(e.g.\ proof) for a formula that had been preprocessed before it was
given to a solver? The paper targets a suite of commonly-used
preprocessing techniques and shows how to reconstruct certificates for them.
On the negative side, the paper discusses certain limitations of the
currently-used proof systems in the light of preprocessing. The
presented techniques were implemented and evaluated in the
state-of-the-art QBF preprocessor~bloqqer.

\end{abstract}
\section{Introduction}


Preprocessing~\cite{DBLP:conf/sat/GiunchigliaMN10,DBLP:conf/cp/SamulowitzDB06,DBLP:conf/sat/SamulowitzB06,DBLP:conf/cade/BiereLS11} and certificate generation~\cite{DBLP:conf/lpar/Benedetti04,DBLP:conf/cade/Benedetti05,DBLP:conf/sat/JussilaBSKW07,DBLP:series/faia/GiunchigliaMN09,DBLP:journals/aicom/NarizzanoPPT09,DBLP:conf/sat/NiemetzPLSB12} are both active areas of research related to QBF solving.
Preprocessing makes it possible to solve many more problem instances.
Certification ensures results are correct, and certificates are themselves useful in applications.
In this paper we show how to generate certificates while preprocessing is used.
Hence, it is now possible to certify the answers for many more problem instances than before.

QBF solvers are practical tools that address the standard
PSPACE-complete problem: given a closed QBF, decide whether it is true.
In principle, such solvers can be applied to any PSPACE problem, of which there are many; for example, model checking in first-order logic~\cite{stockmeyer1974}, satisfiability of word equations~\cite{DBLP:journals/jacm/Plandowski04}, the decision problem of the existential theory of the reals~\cite{DBLP:conf/stoc/Canny88}, satisfiability for many rank-1 modal logics~\cite{DBLP:journals/tocl/SchroderP09}, and so on~\cite{DBLP:series/faia/GiunchigliaMN09,DBLP:journals/jsat/BenedettiM08,DBLP:journals/entcs/JussilaB07}.
Unlike SAT solvers (for NP problems), QBF solvers are not yet routinely used in practice to solve PSPACE problems: they need to improve.

Fortunately, QBF solvers do improve rapidly~\cite{qbfgallery2013}.
One of the main findings is that a two-phase approach increases considerably the number of
instances that can be solved in practice: in the first phase, \emph{preprocessing}, a range of
fast techniques is used to simplify the formula; in the second phase, actual solving, a
complete search is performed.
Another recent improvement is that QBF solvers now produce \emph{certificates}, which include the true\slash false answer together with a justification for it.
Such a justification can be for example in the form of a proof of the given formula.
Certificates ensure that answers are correct, and are sometimes necessary for other reasons.
For example, certificates are used to suggest repairs in QBF-based diagnosis~\cite{ng_thesis,DBLP:conf/sat/StaberB07,DBLP:conf/fmcad/SamantaDE08}.

Clearly, both preprocessing and certificate generation are desirable.
Alas, no tool-chain supports both preprocessing and certificate generation at the same time.
This paper shows how to reconstruct certificates in the presence of a wide range of preprocessing techniques.
In our setup (\autoref{fig:architecture}), the preprocessor produces a simplified formula together with a \emph{trace}.
After solving, we add a postprocessing step, which uses the trace to reconstruct a certificate for the original formula out of a certificate for the simplified formula.

\begin{figure}[t]\centering\scriptsize
\begin{tikzpicture}[xscale=4,yscale=.8]
\tikzset{code/.style={rectangle,draw,text width=11em,text height=2ex,text depth=0.5ex,fill=yellow!20,align=center}}
\tikzset{data/.style={code,rounded corners=5pt,fill=black!10}}
\tikzset{arr/.style={thick,->}}
\node (ans1) at (0,6) {correct/incorrect};
\node (ans2) at (0,0) {correct/incorrect};
\node[code] (check1) at (0,5) {checker};
\node[code] (prepro) at (-1,3) {preprocessor};
\node[code] (backpo) at (1,3) {postprocessor};
\node[code] (solver) at (0,2) {solver};
\node[code] (check2) at (0,1) {checker};
\node[data] (bigq) at (-1,4) {original QBF};
\node[data] (bigc) at (1,4) {reconstructed certificate};
\node[data] (trace) at (0,3) {trace};
\node[data] (smallq) at (-1,2) {simplified QBF};
\node[data] (smallc) at (1,2) {simple certificate};
\draw[arr] (check1) -- (ans1);
\draw[arr] (bigq) |- (check1); \draw[arr] (bigc) |- (check1);
\draw[arr] (bigq) -- (prepro); \draw[arr] (backpo)--(bigc);
\draw[arr] (prepro) -- (trace); \draw[arr] (trace) -- (backpo);
\draw[arr] (prepro) -- (smallq); \draw[arr] (smallc) -- (backpo);
\draw[arr] (smallq) -- (solver); \draw[arr] (solver) -- (smallc);
\draw[arr] (smallq) |- (check2); \draw[arr] (smallc) |- (check2);
\draw[arr] (check2)--(ans2);
\end{tikzpicture}
\caption{Architecture}
\label{fig:architecture}
\end{figure}
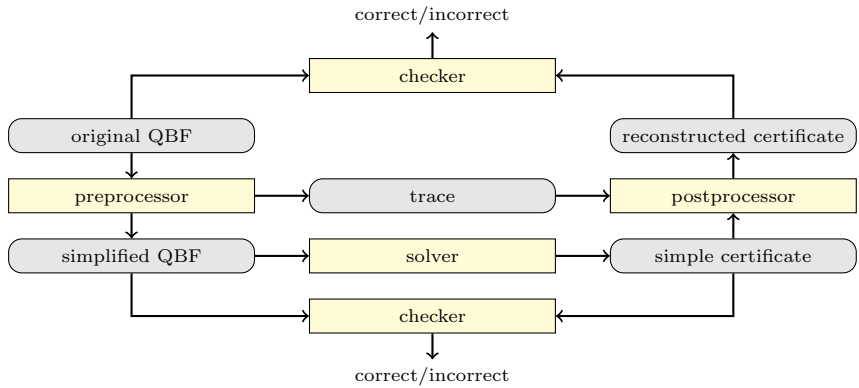

The contributions of this paper are the following:
\begin{itemize}
\item a review of many preprocessing techniques used in practice (\autoref{section:techniques})
\item a negative result about the reconstruction of term resolution-based certificates
(\autoref{section:limitations})
\item certificate reconstruction techniques, corresponding to a wide range of formula preprocessing techniques (\autoref{section:reconstruction})
\item an efficient implementation, and its experimental evaluation (\autoref{section:experiments})
\end{itemize}


\section{Preliminaries}

A {\em literal\/} is a Boolean variable or its negation.  For a literal
$ l $, we write $\bar l $ to denote the literal {\em complementary\/} to
$ l $, i.e.\ $\bar x=\lnot x$ and $\overline{\lnot x}=x$; we write $\vars(l)$ for $x$.
A {\em clause\/} is a disjunction of literals.  A formula
in {\em conjunctive normal form\/} (CNF) is a conjunction of
clauses. Whenever convenient, a clause is treated as a set
of literals, and a CNF formula as a set of sets of clauses.  Dually to
a clause, a {\em term\/} is a conjunction of literals.  A
formula in {\em disjunctive normal form\/} (DNF) is a conjunction of
terms.

For a set of variables $X$, an {\em assignment\/} $\tau$ is a function from $X$
to the constants $0$~and~$1$. We say that $\tau$ is {\em
  complete\/} for $X$ if the function is total.



{\em Substitutions\/} are denoted as $\psi_1/x_1,\dots,\psi_n/x_n$, with
$x_i\neq x_j$ for $i\neq j$. An
application of a substitution is denoted as
$\phi[\psi_1/x_1,\dots,\psi_n/x_n]$ meaning that variables~$x_i$ are
simultaneously substituted with corresponding formula~$\psi_i$ in~$\phi$.

%

{\em Quantified Boolean Formulas\/}
(QBFs)~\cite{DBLP:series/faia/BuningB09} extend propositional logic
with quantifiers that have the standard semantics: $\forall x.\,\Psi$
is satisfied by the same truth assignments as
$\Psi[0/x]\land\Psi[1/x]$, and $\exists x.\,\Psi$ as
$\Psi[0/x]\lor\Psi[1/x]$.
 Unless specified otherwise,   QBFs are in
{\em  closed prenex\/}  form, 
i.e.\ in the form
\hbox{${\cal Q}_1 x_1 \dots {\cal Q}_k x_k.\,\phi$},
where
$x_i$ form a nonrepeating sequence of variables and
${\cal Q}_i\in\{\exists,\forall\}$;
 the formula $\phi$   is  over the variables  $\{x_1,\dots,x_k\}$.
 The propositional part $\phi $ is called the {\em matrix\/} and the rest the {\em prefix}.
If additionally the matrix is in CNF, we say that the formula is in QCNF.
%
A prefix $\Pref$ induces {\em ordering on
  literals\/}~\cite{Buning:1999}: for literals $l_1$,
$l_2$ we write $l_1<l_2$ and say that $l_1$ is {\em less than\/} $l_2$
if $\vars(l_1)$ appears before $\vars(l_2)$ in $\Pref$.

A closed QBF is {\em false\/} (resp.\ {\em true\/}), iff
it is semantically equivalent to the constant $0$ (resp.\ $1$).
 If a variable is universally quantified, we say that the variable is
 \emph{universal}.  For a literal $l$ and a universal variable $x$ such
 that $\vars(l)=x$, we say that~$l$ is \emph{universal}.
 Existential variable and literal are defined analogously.

\subsection{QU-resolution}\label{sec:Qresolution}

{\em QU-resolution\/}~\cite{DBLP:conf/cp/Gelder12} is a
calculus for showing that a QCNF is false. It comprises two operations,
{\em resolution\/} and {\em $\forall$-reduction}.  Resolution is defined
for two clauses $C_1\lor x$ and $C_2\lor\bar x$ such that $C_1\union C_2$
does not contain complementary literals nor any of the
literals $x$, $\bar x$. The {\em QU-resolvent\/} (or simply resolvent) of such
clauses is the clause $C_1\lor C_2$.
The {\em $\forall$-reduction\/} operation removes from a clause~$C$
all universal literals~$l$ for which there is {\em no\/} existential literal $k\in C$ s.t.\
$l<k$.

For a QCNF $\Pref.\,\phi$, a {\em QU-resolution proof\/} of a clause $C$
is a finite sequence of clauses $C_1,\dots,C_n$ where $C_n=C$ and any
$C_i$ in the sequence is part of the given matrix $\phi$; or it is a
QU-resolvent for some pair of the preceding clauses; or it was
obtained from one of the preceding clauses by $\forall$-reduction.  A
QU-resolution proof is called a {\em refutation\/} iff $C$ is the empty
clause.

QU-resolution is a slight extension of
Q-resolution~\cite{DBLP:journals/iandc/BuningKF95}. Unlike
QU-resolution, Q-resolution does {\em not\/} enable resolving on universal
literals. While Q-resolution is on its own refutationally complete for QCNF,
resolutions on universal literals are useful in certain situations
(see also~\cite{EglyWidlQBFW13}).

\subsection{Term-Resolution and Model Generation}
\label{sec:tresd}
{\em Term-resolution\/} is analogous to Q-resolution with the difference
that it operates on terms and its purpose is to prove that a QBF is
true~\cite{GiunchigliaEtAlJAIR06}. Resolution is defined for two terms
$T_1\land x$ and $T_1\land\bar x$ where $T_1\union T_2$ do not contain
any complementary literals nor any of the literals $x$, $\bar x$; the
resolvent is the term $T_1\land T_2$. The {\em $\exists$-reduction\/}
operation removes from a term~$T$ all existential literals~$l$ such that there is
{\em no\/} universal literal $k\in T$ with~$l<k$.

Since term-resolution is defined on terms, i.e.\ on DNF, the {\em
  model generation\/} rule is introduced in order to enable generation
of terms from a CNF matrix.  For a QCNF $\Phi = \Pref.\,\phi$, a term $
T$ is generated by the model generation rule if for each clause $C$
there is a literal $l$ s.t.\ $l\in C$ and $l\in T$. Then, a {\em
  term-resolution proof\/} of the term $T_m$ from $\Phi$ a is a finite
sequence $T_1,\dots,T_m$ of terms such that each term $T_i$ was
generated by the model generation rule; or it was obtained from the
previous terms by $\exists$-reduction or term-resolution. Such proof
{\em proves\/} $\Pref.\,\phi$ iff $T_m$ is the empty term.
(Terms are often referred to as `cubes', especially in the
context of DPLL QBF solvers that apply cube learning.) In the
remainder of the article, whenever we talk about term-resolution proofs
for QCNF, we mean the application of the model generation and
term-resolution rule.
A QCNF formula is true iff it has a
term-resolution proof~\cite{GiunchigliaEtAlJAIR06}.

In this paper, both term-resolution and QU-resolution proofs are
treated as connected directed acyclic graphs so that the each clause/term
in the proof corresponds to some node labeled with that clause/term.

\subsection{QBF as Games}
The semantics of QBF can be stated as a game between an {\em
  universal\/} and an {\em existential
  player\/}~\cite{DBLP:books/daglib/0023084}. The universal player
assigns values to universal variables and analogously the existential
player assigns values to the existential variables. A player assigns a
value to a variable if and only if all variables preceding it in the
prefix were assigned a value. The universal player wins if under the
complete resulting assignment the underlying matrix evaluates to false
and the existential player wins if the underlying matrix evaluates to
true.  A formula is true iff there exists a winning strategy for the
existential player.
The notion of strategy was formalized into {\em models of QBF\/}~\cite{DBLP:journals/jar/BuningSZ07}.

\begin{definition}[Strategy and Model]\label{def:model}
  Let $\Phi = \Pref.\,\phi$ be QBF with the universal variables
  $u_1,\dots,u_n$ and with the existential variables $e_1,\dots,e_m$.  A {\em
    strategy} $M$ is a sequence of propositional formulas
  $\psi_{e_1},\dots,\psi_{e_m}$ such that each $\psi_{e_i}$ is over
  the universal variables preceding $e_i$ in the quantification order.
  We refer to the formula $\psi_x$ as the {\em definition\/} of $x$ in $M$.

  A strategy~$M$ is a {\em model\/} of $\Phi$ if and only if the
  following formula is true
    \[\forall u_1,\dots,u_n.\,\phi[\psi_{e_1}/e_1,\dots,\psi_{e_m}/e_m]\]
  i.e.,\ $\phi[\psi_{e_1}/e_1,\dots,\psi_{e_m}/e_m]$ is a tautology.
\end{definition}

\paragraph {\bf Notation.}
{\it  Let $\Phi=\Pref.\,\phi$ be a QBF as in \autoref{def:model} and
  $M=(\psi_{e_1},\dots,\psi_{e_m})$ be a strategy.
  For a formula~$\xi$ we write $M(\xi)$ for the formula
  $\xi[\psi_{e_1}/e_1,\dots,\psi_{e_m}/e_m]$.
   For a total assignment $\tau$ to the universal variables $U =
   u_1,\dots,u_n$, we write $M(\xi,\tau)$ for $M(\xi)[\tau(u_1)/u_1,\dots,\tau(u_n)/u_n]$.
   Intuitively, $M(\xi,\tau)$ is the result of the game under strategy $M$ and the moves $\tau$.
   Hence, if $\xi$ is over the variables of $\Phi$, then $M(\xi)$ is
   over~$U$ and $M(\xi,\tau)$ yields the constant which results from
   evaluating $\xi$ under the strategy $M$ and assignment $\tau$. In
   particular, $M$ is a model of $\Phi$ iff $M(\xi,\tau)=1$ for
   any~$\tau$.}

\begin{example}
  For a QCNF $\forall u\exists e.\, (\bar u\lor e)\land(u\lor\bar e)$,
  the strategy $M = (\phi_e)$, where $\phi_e = u$ is a model.
  Observe that $M(\bar u\lor e)= M(u\lor\bar e)=u\lor\bar u$ are
  tautologies.
\end{example}

A formula QCNF is true if and only if it has a
model~\cite[Lemma~1]{DBLP:journals/jar/BuningSZ07};  deciding whether
a strategy is a model of a formula is
coNP-complete~\cite[Lemma 3]{DBLP:journals/jar/BuningSZ07}.
We should note that here we follow the definition of model by
B\"uning~{\it et.~al.}, which has a syntactic nature. However,
semantic-based definitions of the same concept appear in
literature~\cite
{DBLP:conf/sat/JussilaBSKW07,DBLP:journals/fmsd/BalabanovJ12}.

\section{QBF Preprocessing Techniques}\label{section:techniques}

For the following overview of preprocessing techniques we consider a
QCNF $\Pref.\,\phi$ for some quantifier prefix $\Pref$ and a CNF
matrix $\phi$.  All the techniques are validity-preserving.

Let $C\in\phi$ be a clause comprising a single existential literal~$l$.
{\em Unit propagation\/} is the operation of removing from~$\phi$
all clauses that contain~$l$, and removing the literal~$\bar l$
from clauses containing it.

A clause $C\in\phi$ is {\em subsumed\/} by a different clause $D\in\phi$
if $D\subseteq C$; {\em subsumption removal\/} consists in removing
clause $C$.

Consider clauses $C,D\in\phi$ together with their resolvent $R$.  If
$R$ subsumes $C$, then we say that $C$ is {\em strengthened\/} by
{\em self-subsumption\/} using $D$. {\em Self-subsumption strengthening\/}
consists in replacing $ C $ with $R$~\cite{DBLP:conf/sat/EenB05}.

A literal $l$ is {\em pure\/} in $\Phi$ if $\bar l$ does not appear
in~$\phi$.  If $l$ is pure and universal, then the {\em pure literal
  rule\/} (PRL)~\cite{DBLP:journals/jar/CadoliSGG02} consists in removing all
occurrences of $l$.
If $l$ is pure and existential, then the PLR removes all the clauses
containing the literal $ l $.

The technique of {\em blocked clause
  elimination\/}~(BCE)~\cite{DBLP:journals/tcs/Kullmann99,DBLP:conf/cade/BiereLS11} hinges on the
definition of a {\em blocked literal}.
An existential literal~$l$ is blocked in a clause $C$
if for any clause $D\in\phi$ s.t.\ $\bar l\in D$ there is a
literal $k\in C$ with $k<l$ and $\bar k\in D$.
A clause is blocked if it contains a blocked literal. BCE consists in removing blocked clauses from the
matrix.

{\em Variable elimination\/}~(VE)~\cite{DBLP:conf/cp/PanV04,DBLP:conf/sat/GiunchigliaMN10}
replaces all clauses containing a certain variable with all their possible resolvents
on that variable.
In QBF, to ensure soundness, the technique is carried out only if a
certain side-condition is satisfied.  For an
existential variable $x$, let us partition $\phi$ into
$\phi_x\cup\phi_{\bar x}\cup\xi$ where $\phi_x$~has all clauses
containing the literal~$x$, and $\phi_{\bar x}$~has all clauses
containing the literal~$\bar x$.
For any clause $C\in\phi_x$ that contains some literal $k$ s.t.~$x<k$
and any clause $D\in\phi_{\bar x}$, there is a literal $z<x$ s.t.\ 
$z\in C$ and $\bar z\in D$. Variable elimination consists in
replacing $\phi_x\cup\phi_{\bar x}$ with the set of
resolvents between the pairs of clauses of $\phi_x$ and $\phi_{\bar
  x}$ for which the resolution is defined.


The {\em binary implication
  graph\/}~(e.g.\ \cite{DBLP:conf/lpar/HeuleJB10}) $G_\phi$ is
constructed by generating for each binary clause $l_1\lor l_2\in\phi$
two edges: $\bar l_1\to l_2$ and $\bar l_2\to l_1$.
If two literals appear in the same strongly
connected component of $G_\phi$, then they must be equivalent.  {\em
  Equivalent literal substitution\/}~(ELS) consists in replacing
literals appearing in the same strongly connected component $S$ by one
of the literals from~$S$; this literal is called the {\em
  representative}. The representative is then substituted in place of
the other literals of $S$.
While in plain SAT preprocessing a representative
can be chosen arbitrarily, in QBF it must be done with care. First,
three conditions are checked: (1)~$S$ contains two distinct universal
literals (also covers complementary universal literals); (2)~$S$
contains an existential literal $l_e$ and a universal literal~$l_u$ such
that $l_e<l_u$; (3)~$S$ contains two complementary
existential literals. If either of the conditions (1), (2), or (3) is
satisfied, then the whole formula is
false~(cf.~\cite{DBLP:journals/ipl/AspvallPT79}), and ELS stops.
Otherwise,  ELS picks as representative the literal
that is the outermost with respect to the considered prefix. Observe
that if the component contains exactly one universal literal, it
will be chosen as the representative. All clauses that become
tautologous due to the substitution, are removed from the matrix (this
includes the binary clauses that were used to construct the strongly
connected~components).
%

\section {Limitations}
\label{section:limitations}

 In this section we focus on the limitations of currently-available
 calculi from the perspective of preprocessing. In particular, we show
 that term-resolution+model-generation proofs cannot be tractably
 reconstructed for blocked clause elimination and variable
 elimination.
 For a given parameter $n\in\mathbb{N}^+$ construct the following true QCNF
  with $2n$~variables and $2n$~clauses.


 \begin{equation}\label{equation:iff}
   \textstyle
   \forall u_1\exists e_1%
   \dots%
   \forall u_n\exists e_n.\,
  \bigwedge_{1\le i\le n} 
  \;(\bar u_i\lor e_i)\land (u_i\lor\bar e_i)
 \end{equation}

 \begin{proposition}\label{proposition:term_exponential}
    Any term-resolution proof of~\eqref{equation:iff} has size exponential in~$n$.
 \end{proposition}
 \begin{proof}

   Pick an arbitrary assignment $\tau$ to the universal variables
   $u_1,\dots,u_n$.  We say that a term $T$ {\em agrees} with an
   assignment $\tau$ iff there is no literal $l$ such that $\bar l\in
   T$ and $\tau(l)=1$.  Given a term-resolution proof $\pi $
   for~\eqref{equation:iff}, we show that $\pi$ must have a leaf that
   agrees with $\tau $ by constructing a path from the root to some
   leaf such that each node on that path agrees with $\tau$.  The root
   of $\pi$ agrees with $\tau$ because it does not contain any
   literals. If a term $T$ agrees with $\tau$, and $T$ is obtained from
   $T'$ by $\exists$-reduction, then $T'$ also agrees with $\tau$
   since $\tau $ assigns only to universal variables. If $ T $ agrees
   with $\tau $ and is obtained from $T_0$ and $T_1$ by
   term-resolution on some variable $y$, then $y\in T_k$ and $\bar y\in T_{1-k}$
   for some $k\in\{0,1\}$.
   Hence, at least one of the terms $T_0$~and~$T_1$ agrees with $\tau$.

   Recall that each leaf $T$ of $\pi$ must be obtained by the
   model-generation rule; i.e.\null, for each clause $C$
   of~\eqref{equation:iff} there is a literal~$l$ s.t.~$l\in C$ and
   $l\in T$. Hence, for each pair of clauses $(\bar u_i\lor e_i)\land
   (u_i\lor\bar e_i)$ either $\bar u_i, \bar e_i\in T$ or $u_i,e_i\in T$. 
   Consequently, each leaf of~$\pi$ has $n$ universal literals.

   For each of the $2^n$ possible assignments $\tau$, the proof $\pi$
   must contain a leaf~$T_\tau$ that agrees with~$\tau$. Since
   $T_\tau$ contains $n$ universal literals, for a different
   assignment~$\tau'$ there must be another leaf $T_{\tau'}$ that
   agrees with it. Overall, $\pi$ must must contain at least~$2^n$
   different terms.
 \end{proof}

 \begin{proposition}\label{proposition:bce_ve_polynomial}
    Both blocked clause elimination and variable elimination
    reduce the matrix of~\eqref{equation:iff} to the empty set of
    clauses in polynomial time.
 \end{proposition}
 \begin{proof} Immediate from definitions of blocked clause and variable elimination.

 \end{proof}

\begin{corollary}\label {corollary:limitation}
  If blocked clause elimination or variable elimination are used for
  preprocessing, then reconstructing a term-resolution proof takes
  exponentially more time than preprocessing, in the worst case.
\end{corollary}

In the remainder of the paper we do not consider
term-resolution+model-generation proofs for certification since
\autoref{corollary:limitation} shows that, in the context of preprocessing, this
calculus is not appropriate.
Rather than term-resolution, we will use models to certify true
formulas. We should note, however, that for such we are paying a price
of higher complexity for certificate verification. While
term-resolution+model-generation proofs can be verified in polynomial
time, verification of models is coNP-complete.  (For false formulas,
QU-resolution is used for certification, which is still verifiable in
polynomial time.)

In a similar spirit, we do not consider the preprocessing technique of
universal-expansion~\cite{BubeckBuningSAT07}, which is based on the identity
$\forall x.\,\Phi=\Phi[1/x]\land\Phi[0/x]$. While there is no hard
evidence that there is no tractable algorithm for reconstructing
QU-resolution proofs for universal-expansion, recent work hints in this
direction~\cite{DBLP:conf/sat/JanotaM13}. Hence, only the techniques
described in \autoref{section:techniques} are considered.

\section{Certificate Reconstruction}
\label{section:reconstruction}
  This section shows how to produce certificates in the context of
  preprocessing. In particular, we focus on two types of certificates:
  QU-resolution refutations (\autoref{sec:Qresolution}) for false
  formulas and models (\autoref{def:model}) for true formulas.  We
  consider each of the techniques presented in
  \autoref{section:techniques} and we show how a certificate is {\em
    reconstructed\/} from the certificate of the preprocessed
  formula. This means that reconstruction produces a model
  (resp.\ refutation) for a formula $\Phi$ from a model
  (resp.\ refutation) for a formula $\Phi'$, which resulted from $\Phi$
   by the considered technique.  For nontrivial
  reconstructions we also provide a proof of why the reconstruction is
  correct.

  Having a reconstruction for each of the preprocessing techniques
  individually enables us to reconstruct a certificate for the whole
  preprocessing process. The preprocessing process produces a
  sequence of formulas $\Phi_0,\dots,\Phi_n$ where $\Phi_0$ is the
  input formula, $\Phi_n$ is the final result,
  and each formula $\Phi_{i+1}$ is obtained from $\Phi_i$ by \emph{one}
  preprocessing technique.  For the purpose of the reconstruction, we
  are given a certificate $\Cert_n$ for the formula $\Phi_n$. This
  final certificate $\Cert_n$ is in practice obtained by a QBF
  solver.  The reconstruction for the whole processing process works
  backwards through the sequence of formulas $\Phi_0,\dots,\Phi_n$.
  Using $\Cert_n$, it reconstructs a certificate $\Cert_{n-1}$ for the
  formula $\Phi_{n-1}$, then for $\Phi_{n-2}$ and so on until it
  produces a certificate $\Cert_0$ for the input formula. The reminder
  of the section describes these individual reconstructions for the
  considered techniques.

  We begin by two simple observations. If a transformation removes a
  clause, then reconstruction of a QU-resolution proof does not need to do
  anything. Analogously, reconstruction of models is trivial for
  transformations adding new~clauses.
  \begin{observation}
  \label{observation:clause_removal}
  Consider a QCNF $\Phi=\Pref.\,\phi $ and a clause $C\in\phi$.  Any
  QU-resolution proof of $\Phi'=\Pref.\,\phi\smallsetminus\{C\}$ is
  also a QU-resolution proof of $\Phi$.
  \end{observation}

  \begin{observation}
  \label{observation:clause_addition}
  Consider a QCNF $\Phi=\Pref.\,\phi $ and a clause $C$ over the
  variables of $\Phi$.  Any model of $\Phi'=\Pref.\,\phi\union\{C\}$
  is a model of $\Phi$.
  \end{observation}

\subsection{Subsumption, Self-Subsumption, and Unit Propagation}
\label{section:basic}

In the case of subsumption, a QCNF $\Phi = \Pref.\,\phi$ is
transformed into $\Phi' = \Pref.\,\phi\smallsetminus\{C\}$ for a clause
$C$ for which that there is another clause $D\in\phi$ such that
$D\subseteq C$. For reconstructing QU-resolution nothing needs to be
done due to \autoref{observation:clause_removal}. For any model $M'$
of $\Phi'$, the formula $M'(\phi\smallsetminus\{C\})$ is a tautology and in
particular $M'(D)$ is a tautology and therefore necessarily $M'(C)$ is
a tautology because $C$~is weaker than~$D$. Hence, $M'(\phi)$ is a tautology
and $M'$ is also a model of $\Phi$.

In order to reconstruct unit propagation and self-subsumption we
first show how to reconstruct resolution steps. For such, consider the
transformation of a QCNF $\Phi = \Pref.\,\phi$ into the formula
$\Phi'=\Pref.\,\phi\union\{C\}$ where $C$ is a resolvent of some
clauses $D_1,D_2\in\phi$. Any QU-resolution proof $\pi'$ of $\Phi'$ where $C$
appears as a leaf of $\pi'$ is transformed into a QU-resolution proof of $\Phi$
by prepending this leaf with the resolution step of $D_1$ and $D_2$.
Any $M'$ model of $\Phi' $ is also a model of $\Phi$ due to
\autoref{observation:clause_addition}.

Each self-subsumption strengthening consists of two steps: resolution and
subsumption. Unit propagation consists of resolution steps,
subsumption, and the pure literal rule (see \autoref{section:pureliterals}).
Hence, certificates are reconstructed accordingly. Note that
in self-subsumption strengthening, resolution steps may be carried out
on universal literals while in unit propagation this would not be
meaningful because the moment the matrix contains a unit clause where
the literal is universal, the whole formula is trivially false due to
universal reduction.
%

\subsection{Variable Elimination (VE)}
To eliminate a variable $ x $ from $\Pref.\,\phi$, VE partitions the
matrix $\phi$ into the sets of clauses $\phi_x$, $\phi_{\bar x}$, and
$\xi$ as described in \autoref{section:techniques}.
Subsequently, $\phi_x$ and $\phi_{\bar x}$ are replaced by the set $\phi_x\otimes\phi_{\bar
  x}$, which is defined as the set of all possible resolvents on~$x$
of clauses that do not contain another complementary literal.  Recall
that VE can be only carried out if the side-condition specified
in \autoref{section:techniques} is fulfilled.

To reconstruct a QU-resolution proof we observe that VE can be split
into operations already  covered.  The newly added clauses are
results of resolution on existing clauses, which was already covered
in \autoref{section:basic}. Clauses containing $x$ are removed, which
does not incur any reconstruction due to
\autoref{observation:clause_removal}.

To reconstruct models we observe that any given formula $\Phi$ 
can be written as
$\Phi=\Pref_1\exists x\Pref_2.\,(x\lor\phi_1)\land(\bar x\lor\phi_2)\land\xi$ 
for CNF formulas $\phi_1$, $\phi_2$, and $\xi$ that do not contain~$x$.
Then, VE consists in
transforming $\Phi$ into the formula
$\Phi'=\Pref_1\Pref_2.\,(\phi_1\lor\phi_2)\land\xi$
(note that $\phi_1\lor\phi_2$ corresponds to $(x\lor\phi_1)\otimes(\bar
x\lor\phi_2)$).
VE's side-condition
specifies that any clause $C\in\phi_1$ that contains some literal~$k$ such that
$k>x$ and any clause $D\in\phi_2$, there is a literal $z<x$ such that
$z\in C$ and $\bar z\in D$.

In order to construct a model for the original formula $\Phi$ from a
model~$M'$ of~$\Phi'$, we aim to add to $M'$ a definition for $x$
which sets~$x$ to~$1$ when $\phi_1$ becomes~$0$ and it sets it to~$1$
when $\phi_2$ becomes~$0$. Since $M'$ is a model of $\Phi'$, the
strategy $M'$ satisfies one of the $\phi_1$,~$\phi_2$ for any game.
The difficulty lies in the fact that $\phi_1$ and $\phi_2$ may contain
variables that are on the right from $x$ in the quantifier prefix (those in
$\Pref_2$) and these must not appear in the definition of~$x$. Hence,
we cannot use $\phi_1$ and $\phi_2$ to define~$x$ as they are.  Instead,
we construct a formula~$\phi_2'$ by removing from $\phi_2$ all
unsuitable literals, i.e.\ literals $k$ for which $x<k$.  Then, we set
the definition for $x$ to $M'(\phi'_2)$.  Now whenever $\phi'_2$
evaluates to 1,  so do $\phi_2$ and $(x\lor\phi_1)\land(\bar
x\lor\phi_2)$, because $x$ is set to~$1$.  If, however,
$\phi'_2$ evaluates to~$0$, then $\phi_2$ might not necessarily evaluate to~$0$,
but $x$ is set to~$0$ by our strategy regardless.  Due to the
side-condition, in such cases $\phi_1$ must evaluate to~$1$ and
therefore our strategy is safe. This is formalized by the following
proposition.

%

\begin{proposition}
Let $\Phi=\Pref_1\exists x\Pref_2.\,(x\lor\phi_1)\land(\bar x\lor\phi_2)\land\xi$
with $\phi_1$~and~$\phi_2$ not containing~$x$;
let $\Phi' =\Pref_1\Pref_2.\,(\phi_1\lor\phi_2)\land\xi$, as above.
Define $\phi_2'$ to be $\phi_2$ with all the literals not less than $x$ deleted;
i.e.,
$\phi_2'=\comprehension{\comprehension{l}{l\in C,l<x}}{C\in\phi_2}$.
If $M'$~is a model for~$\Phi'$, then $M=M'\cup\{\psi_x\}$ is a model for~$\Phi$,
where $\psi_x=M'(\phi_2')$.
\end{proposition}
\begin{proof}
  The functions of $M$ form a well-defined strategy since $M'$ is
  a well-defined strategy and~$\psi_x$ does not contain any
  literals~$k$ with $k>x$.
  To show that $M$ is a model of $\Phi$, consider any complete
  assignment $\tau$ to the universal variables of $\Phi$.  Now we
  wish to show that the matrix of $\Phi$ evaluates to~$1$ under $M$
  and $\tau$.
  Since $M'$ is a model of $\Phi'$, and $\xi$ does not contain $x$, it
  holds that $M(\xi,\tau)=M'(\xi,\tau)=1$. So it is left to be shown
  that the subformula $(x\lor\phi_1)\land (\bar x\lor\phi_2)$ is true
  under $M$ and $\tau$.

  Because $\phi_1$, $\phi_2$ do not contain $x$ we have
  $M(\phi_1) = M'(\phi_1)$,
  $M(\phi_2) = M'(\phi_2)$, 
  $M(\phi_2') = M'(\phi_2')$, 
  and
  $M(\phi_1\lor\phi_2,\tau)=M'(\phi_1\lor\phi_2,\tau)=1$.
  Split on the following cases (distinguishing between the values of~$x$ 
   under $\tau$ and $M$).

          If $M(x,\tau)=M'(\phi_2',\tau)=1$.
          Because $\phi_2'$ is stronger than $\phi_2$, i.e.\ $M(\phi_2')\limp M(\phi_2)$,
          also $M(\phi_2,\tau)=1$. Hence $M((x\lor\phi_1) \land (\bar x\lor \phi_2),\tau) = 1$.

         If $M(x,\tau)=M'(\phi_2',\tau)=0$. 
         There must be a clause $C'\in\phi_2'$ s.t.\ $M'(C',\tau)=0$, 
         i.e.\ for all literals $l\in C'$, $M'(l,\tau)=0$.
         Let $C\in\phi_2$ be a clause from which $C'$ 
         resulted by removing some literals (possibly none), 
         i.e.\ $C'=\comprehension{l}{l\in C, l<x}$.
         Now consider two sub-cases depending on whether $C=C'$ or $C\neq C'$.
         If $C=C'$, $M'(C,\tau)=0$ and $M'(\phi_2,\tau)=0$, from which
         $M'(\phi_1,\tau)=1$ because $M'(\phi_1\lor\phi_2,\tau)=1$.
         Hence $M((x\lor\phi_1)\land(\bar x\lor\phi_2)) = 1$.
         If $C\neq C'$, due to the side-condition, $C$
         contains for each clause $D\in\phi_1$ a literal~$l_D$
         s.t.\ $\bar l_D \in D$ and $l_D<x$.
          Since each literal~$l_D$ is less than $x$, it is also in $C'$.
          Since $M(C',\tau)=0$, each $M(l_D,\tau)=0$ and $M(\bar l_D,\tau)=1$.
          From which $M(\phi_1,\tau)=1$ and $M((x\lor\phi_1) \land (\bar x\lor\phi_2),\tau) = 1$.
\qed
\end{proof}

\subsection {Pure Literal Rule (PLR)}\label{section:pureliterals}
PLR for existential literals is a special case of
both variable elimination and blocked clause elimination. (An
existential pure literal is a blocked literal in any clause.)
Hence, certificate reconstruction for existential PLR is done accordingly.

For a universal literal~$l$ with $\vars(l)=y$, a QCNF $\Phi=\Pref_1\forall
y\Pref_2.\,\phi$ is translated into the QCNF formula
$\Phi'=\Pref_1\Pref_2.\,\phi'$ by removing $l$ from all clauses where
it appears. To obtain a QU-resolution proof $\pi$ for $\Phi$ from a
QU-resolution proof $\pi'$ one inserts $l$ in any of the leafs
$C'\in\phi'$ of $\pi'$ s.t.\ there exists $C\in\phi$ with
$C'=C\smallsetminus\{l\}$.  Then, $\forall$-reductions of~$l$ are
added to~$\pi'$ whenever possible.  Note that the addition of~$l$ cannot lead to
tautologous resolvents since only~$l$ is inserted and never $\bar l$. The
newly added universal literals must be necessarily $\forall$-reduced
as~$\pi'$ eventually resolves away all existential literals.
Since~$l$ is universal, any model of $\Phi'$ is also a model of~$\Phi$.

\subsection{Blocked Clause Elimination (BCE)}

For a QCNF $\Phi=\Pref.\,\phi$, BCE identifies a blocked clause
$C\in\phi$ and a blocked existential literal $l\in C$, and removes $C$
from~$\phi$. Recall that for a blocked literal it holds that for any
$D\in\phi$ such that $\bar l\in D$ there exists a literal $k\in C$ such that $\bar
k\in D$ and~$k<l$.

To reconstruct QU-resolution proofs, nothing needs to be done due to
\autoref{observation:clause_removal}.
To show how to reconstruct models, let $M'$ be a model
for $\Phi'=\Pref.\,\phi\smallsetminus\{C\}$.  Let $W$ be the set of
literals that serve as witnesses for $l$ being blocked, i.e.\
$W = \comprehension{ k\in C}{ k\neq l\text{ and there exists a } D\in\phi\text{ s.t.\ } \bar k,\bar l\in D \text{ and } k<l}$.

The intuition for constructing a model for $\Pref.\,\phi$ is to play the same as $M'$ except
for the case when the literals $W$ are all $0$, then make sure that $l$ evaluates to $1$.
This is formalized by the following proposition.
\begin{proposition}
 Let $\Phi$, $\Phi'$, $M'$, and $W$ be defined as above.
   Let $x=\vars(l)$ and $\psi'_x\in M'$ be the definition for~$x$.
  Define $\psi_x = \psi'_x\lor M'(\bigwedge_{k\in W}\bar k)$ if $l=x$ and
  $\psi_x = \psi'_x\land M'(\bigvee_{k\in W} k)$ if $l=\bar x$. Finally, define
  $M = M'\smallsetminus\{\psi'_x\}\union\{\psi_x\}$. Then $M$ is a model of~$\Phi$.
  (Note that universal literals of $W$ are untouched by~$M'$.)
\end{proposition}
\begin{proof}
Strategy $M$ is well-defined because literals in $W$ are all less
 than $l$ and therefore definitions for those literals
also contains literals less than~$l$.
%
%
Let us consider some total assignment $\tau$ to the universal
variables of $\Phi$ under which all literals in $W$ are $0$ under $M$
(for other assignments $M$ behaves as $M'$ and $C$ is true).  Now let
us split the clauses of $\phi$ into 3 groups.
Clauses that do not contain $\bar l$ nor $l$;
clauses that contain $l$;
and those that contain $\bar l$.
For any clause $D\in\phi$ not containing $l$ nor $\bar l$,
$M(D,\tau)=1$ since $M(D,\tau)=M'(D,\tau)$ and $M'$ is a model of
$\Phi'$.
For any clause $D\in\phi$ containing $l$, $M(D,\tau)=1$
since $M(l,\tau)=1$; this includes the clause $C$.
Due to the sidecondition, any clause $D\in\phi$ that contains $\bar l$
also contains a literal $k$ s.t.\ $\bar k\in W$. Since for $M(\bar
k,\tau)=0$, i.e.\ $M(k,\tau)=1$, it holds that $M(D,\tau)=1$.
%
\qed
\end{proof}




\subsection {Equivalent Literal Substitution (ELS)}

For a formula $\Phi = \Pref.\,\phi$, ELS constructs strongly connected
components of the binary implication graph $G$ of $\phi$.  Once a
strongly connected component $ S $ of the graph is constructed, ELS
checks whether $S$ yields falsity.  If it does, ELS produces a
QU-resolution proof for such. The following discusses scenarios of
falsity that may arise. First recall that if there is a path in $G$ from a
literal $l_1$ to $l_k$ then there is a set of clauses $(\bar l_1\lor l_2$),
$(\bar l_2\lor l_3)$, $\dots$, $(\bar l_{k-1}\lor l_k$), which
through a series of QU-resolution steps enables us to derive the clause
$\bar l_1\lor l_k$. Also recall that whenever there is a path from
$l_1$ to $l_k$ in some component $S_1$, there is also a path from
$\bar l_1$ to $\bar l_k$ in the component $S_2$, obtained from
$S_1$ by negating all literals and reversing all edges.
These observations are repeatedly used in the following~text.

(1)~If $S$ contains two universal literals $l_1$ and $l_2$, derive the
clause $\bar l_1\lor l_2$, which is then $\forall$-reduced to the
empty clause. (Note that this also covers $l_2=\bar l_1$.)

(2)~If~$S$ contains an existential literal~$l_e$ and an universal
literal~$l_u$ such that $l_e<l_u$, derive the clause
$\bar l_e\lor l_u$ from which $\forall$-reduction gives~$\bar l_e$.
Derive $l_e$ analogously. Finally resolve $\bar l_e$ and
$l_e$ to obtain the empty clause.

(3)~If $S$ contains two literals $e$ and $\bar e$ for some existential
variable $e$, derive the unit clauses $e$ and $\bar e$ and resolve them
into the empty clause.

If none of the three conditions above are satisfied, all literals in
$S$ are substituted by a representative literal $r$, which is the
smallest literal from $S$ w.r.t.\ the literal ordering~$<$. This
yields a formula $\Phi'=\Pref'.\,\phi'$, where $\Pref'$ resulted from
$\Pref $ by removing all variables that appear in $S$ except for
$\vars(r)$.  A certificate is reconstructed as follows.

If a QU-resolution proof $\pi'$ for $\Phi'$ relies on a clause $C'\in\phi'$ that
resulted from some cause $C\in\phi $ by replacing a $l\in S$ by $r$,
construct the clause $\bar l\lor r$ and resolve it with~$C$ to
obtain~$C'$. Analogously, if $C'$ resulted from $C$ by
replacing $\bar l\in S$ with $\bar r$, construct the clause $l\lor\bar r$ 
and resolve it with $C$ to obtain $C'$.

If $M'$ is a model of $\Phi'$ and $r$ is existential, then $S$
does not contain any universal literals and $M'$ defines the value for
$r$ by some formula $\psi_r = M'(r)$. In such case $\psi_r$ is over
universal variables that are less than all the
literals in $S$ because $r$ was chosen to be the outermost literal. If
$x\in S$ for some existential variable $x$, set $\psi_x$ as $\psi_r$;
if $\bar x\in S$ for some existential variable $x$, set $\psi_x$ as
$\lnot\psi_r$.
If $r$ is universal, all the other literals in $S$ are existential and so
for $x\in S\smallsetminus \{r\}$ we set $\psi_x=r$;
for $\bar x\in S\smallsetminus\{r\}$,
we set $\psi_x=\bar r$.

\section{Related Work}
Local simplifications based on identities such as $0x=0$ appear in
number of instances of automated reasoning (c.f.~\cite{Harrison09}).
In SAT solving, it was early recognized that going beyond such local
simplifications leads to significant performance gains.  A notable
technique is variable elimination (VE), which originates in the
Davis\&Putnam procedure (DP). While DP is itself complete, it suffers
from unwieldy memory consumption.
It has been shown that applying VE {\em only\/} if it does not lead to
increase of the formula's size, gives an incomplete yet powerful
technique~\cite{DBLP:conf/sat/SubbarayanP04}.
The preprocessor~\satelite~\cite{DBLP:conf/sat/EenB05} boosts VE by
subsumption, self-subsumption, and unit propagation.

Nowadays, preprocessors (and SAT solvers themselves) contain a number
of preprocessing techniques such as {\em blocked clauses
  elimination\/}~\cite{DBLP:journals/tcs/Kullmann99,DBLP:conf/cp/OstrowskiGMS02,DBLP:journals/jar/JarvisaloBH12},
{\em hyper binary resolution\/}~\cite{DBLP:conf/sat/BacchusW03} and
others (cf.~\cite{DBLP:conf/lpar/HeuleJB10}).
%
Reconstructing solutions in SAT is generally easier than in QBF, but it
has also been investigated~\cite{DBLP:conf/sat/JarvisaloB10}.

Many SAT preprocessing techniques were generalized for
QBF~\cite{DBLP:conf/sat/Biere04,DBLP:conf/cp/SamulowitzDB06,DBLP:conf/sat/GiunchigliaMN10,DBLP:phd/de/Bubeck2010,DBLP:conf/cade/BiereLS11};
application thereof is crucial for QBF solving~\cite{qbfgallery2013}.  QBF leads to a number of specifics in
the techniques.  VE can be only performed under a
certain side-condition (\autoref {section:techniques}); Van
Gelder~\cite{DBLP:conf/cp/Gelder12} further generalizes this
side-condition.  A technique specific to QBF is {\em universal-variable
  expansion\/}~\cite{BubeckBuningSAT07,DBLP:phd/de/Bubeck2010}
where a universal quantifier $\forall x.\,\Phi$ is expanded into
$\Phi[0/x]\land\Phi[1/x]$ and then brought into the prenex form by
variable renaming. (Expansion can be used to obtain a complete
solver~\cite{DBLP:conf/lpar/Benedetti04,DBLP:conf/sat/Biere04,DBLP:conf/sat/LonsingB08,DBLP:conf/sat/JanotaKMC12}.)
%
In his recent work, Van Gelder provides some initial insights into
reconstruction of variable elimination and expansion~\cite{AVGQBFW13}.
There, however, he only shows how to reconstruct an individual leaf
of a term-resolution proof, but does not show how to construct the proofs
themselves.

A number of works focus on the certification of QBF solvers
(e.g.~\cite%
{DBLP:conf/cade/Benedetti05%
,DBLP:conf/sat/JussilaBSKW07%
,DBLP:journals/aicom/NarizzanoPPT09%
,DBLP:conf/ijcai/GoultiaevaGB11%
})
motivated by error prevention~\cite{DBLP:conf/sat/BrummayerLB10}, but
also because the certificates themselves can be useful
(e.g.~\cite%
{ng_thesis%
,DBLP:conf/sat/StaberB07%
,DBLP:conf/fmcad/SamantaDE08%
,DBLP:journals/fmsd/BalabanovJ12%
,DBLP:conf/sat/JordanK13}).

\section {Experimental Evaluation}
\label{section:experiments}

We test five scenarios, corresponding to different settings for preprocessing ({\bf f}ull, {\bf s}imple, or {\bf n}one) and for solving (with a {\bf q}dag dependency manager, or {\bf s}imple).
\autoref{table:solved} defines and names the scenarios that we tested\,---\,%
the last letter indicates whether certificate generation was enabled ({\bf y}es or {\bf n}o).
The scenario \emph{nsy} represents the state-of-the-art in QBF solving \emph{with} certificate generation, and is the scenario we set out to improve.
The scenario \emph{ssy} represents our contribution to QBF solving with certificate generation.
We use the QBC format for certificates~\cite{DBLP:conf/sat/JussilaBSKW07}: the size of models is the number of $\land$-gates used, the size of refutations is the number of resolution steps used.
(See online\footnote{\url{http://sat.inesc-id.pt/~mikolas/lpar13-prepro/}} for the exact testing environment being used.)

\paragraph{Results and Discussion.}

\autoref{fig:cactus} shows the overall performance of five scenarios on the QBFEVAL~2012 benchmark.
There is a clear gap between scenarios that use preprocessing (fqn, ssn, ssy) and scenarios that do not use preprocessing (nsn, nsy)\,---\,preprocessing is clearly beneficial.
The gap nsy--nsn shows that enabling tracing in {\sf depqbf} deteriorates its performance.
The gap ssy--ssn is smaller than the gap nsy--nsn, indicating that enabling tracing in {\sf bloqqer}+{\sf depqbf} deteriorates performance \emph{less} than it does for {\sf depqbf} alone.
The gap fqn--ssn should be reduced by future work.
The most important observation to make on \autoref{fig:cactus} is that our proposed scenario (ssy) significantly improves the state-of-the-art in QBF solving with certificate generation~(nsy).
\autoref{table:solved} gives the total number of solved instances for each scenario, thus it corresponds to the rightmost points in \autoref{fig:cactus}.
The generated certificates (in scenarios nsy, ssy) were not all checked:
Those instances on which the certificate checker timed out are listed in the unchecked column.
(Recall that checking strategies is coNP-complete.)
The $7$~unchecked certificates in the nsy scenario are largely disjoint from the $8$~unchecked certificates in the ssy scenario\,---\,the overlap is exactly one instance.

\begin{figure}[t]\centering
\includegraphics{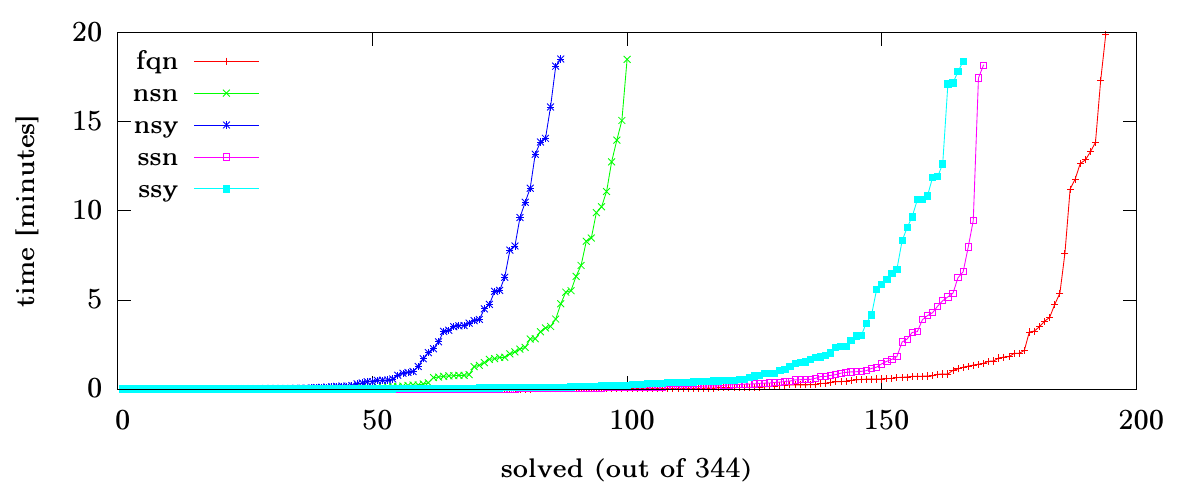}
\caption{Overall performance on the QBFEVAL~2012 benchmark}
\label{fig:cactus}
\end{figure}
\begin{table}[t]\centering
\caption{Number of solved instances out of~$344$, for several scenarios}
\label{table:solved}
\setlength{\tabcolsep}{.5em}
\smallskip
\scriptsize
\begin{tabular}{@{}llllrrrrr@{}}
\toprule
\multicolumn{4}{c}{Scenario}
  & \multicolumn{2}{c}{True/SAT}
  & \multicolumn{2}{c}{False/UNSAT} \\
\cmidrule(r){1-4} \cmidrule(lr){5-6} \cmidrule(lr){7-8}
Name & Preprocessing & Solving & Tracing & Unchecked & Checked & Unchecked & Checked & Total \\ \midrule
fqn & full    & qdag    & no  & 99 & n/a & 94 & n/a & 194 \\
nsn & none    & simple  & no  & 42 & n/a & 58 & n/a & 100 \\
nsy & none    & simple  & yes &  7 &  25 &  0 &  55 &  87 \\
ssn & simple  & simple  & no  & 80 & n/a & 90 & n/a & 170 \\
ssy & simple  & simple  & yes &  8 &  69 &  0 &  89 & 166 \\ \bottomrule
\end{tabular}
\end{table}

\autoref{fig:threshold} shows that preprocessing is beneficial mostly for hard instances.
\autoref{fig:threshold.size} depicts certificate size with preprocessing (ssy) versus certificate size without preprocessing (nsy).
There is a clear threshold around~$10^5$: above it preprocessing helps, below it preprocessing is detrimental.
\autoref{fig:threshold.time} depicts time spent in the solver versus total solving time (which includes preprocessing and postprocessing) for the three scenarios that use preprocessing.
There is a clear threshold around $2$~minutes: above it, scenarios that do not generate certificates (fqn, ssn) have negligible overhead.

\begin{figure}[t]\centering
\begin{subfigure}{6cm}\centering
\includegraphics{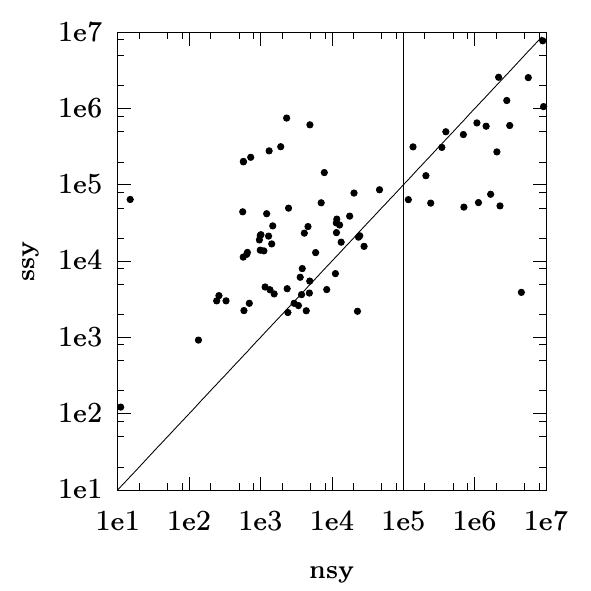}
\caption{Certificate size}
\label{fig:threshold.size}
\end{subfigure}%
\hfil%
\begin{subfigure}{6cm}\centering
\includegraphics{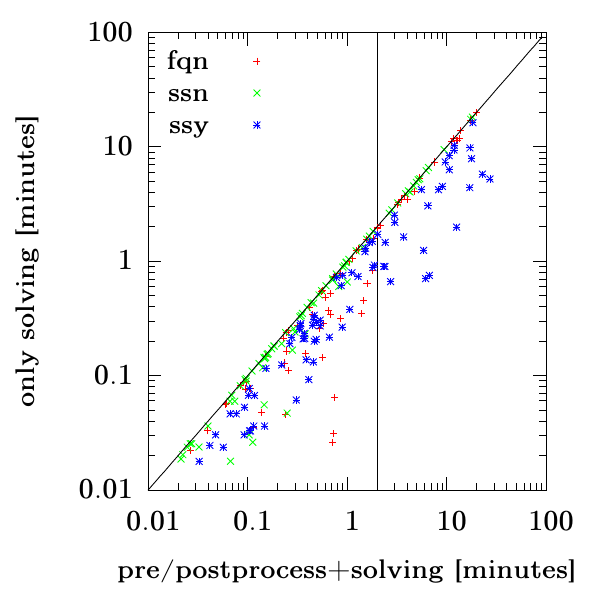}
\caption{Solving time}
\label{fig:threshold.time}
\end{subfigure}
\caption{The effect of pre/postprocessing on certificate size and on solving time}
\label{fig:threshold}
\end{figure}
\begin{table}[t]\centering
\caption{Time spent in solver as a percent of the total solving time.}
\label{table:overheads}
\setlength{\tabcolsep}{.5em}
\smallskip
\def\.{\phantom{~[\%]}}
\begin{tabular}{@{}lrrrr@{}}
\toprule
Scenario & min [\%] & med [\%] & geom avg [\%] & max [\%] \\ \midrule
fqn    &   4\. &  91\. &  66\. & 100\. \\
ssn    &  19\. &  98\. &  86\. & 100\. \\
ssy    &  11\. &  57\. &  50\. &  92\. \\
\bottomrule
\end{tabular}
\end{table}

The correlation between certificate size and total running time is only moderate~($\approx0.6$).
As an example of the high variance, for the 10~instances that were solved in $64$-to-$128$~seconds, the average certificate size was $4.7\times10^5$, with a standard deviation of~$4.8\times10^5$.

\section {Conclusions and Future Work}

This paper brings together two different facets of QBF solving:
preprocessing and certification. Certification is important for
practical applications of QBF and preprocessing is crucial for
performance of nowadays QBF solvers. Both of the facets were
extensively investigated%
~\cite{GiunchigliaEtAlJAIR06,DBLP:conf/cp/Gelder12,DBLP:conf/sat/NiemetzPLSB12,%
  DBLP:conf/sat/Biere04,DBLP:conf/cp/SamulowitzDB06,DBLP:conf/cade/BiereLS11,DBLP:conf/sat/GiunchigliaMN10}
but there is no available toolchain combining the two. However, the
need for such technology has been recognized by others~\cite{qbfgallery2013}.
This paper addresses exactly this
deficiency. For a number of representative preprocessing techniques,
the paper shows how certificates can be reconstructed from a
certificate of a preprocessed formula. Experimental evaluation of the
implemented prototype demonstrates that the proposed
techniques enable QBF solving \emph{with}
certification that is performance-wise very close to a
state-of-the-art QBF solving \emph{without} certification.  Hence, the
contribution of the paper is not only theoretical but also practical
since the implemented tool will be useful to the QBF community.

On the negative side, the paper demonstrates that current methods of
QBF certification are insufficient for full-fledged preprocessing
in the case of true formulas. Namely, term-resolution+model-generation
proofs incur worst-case exponential blowup in blocked clause
elimination and variable elimination. This is an important drawback
because term-resolution proofs can be checked in polynomial time,
which is not the case for model-based certification (used in the
paper). This drawback delimits one direction for future work: Can we
produce polynomially-verifiable certificates for true QBFs in the
context of preprocessing? Another item of future work is narrowing the
performance gap between solving with and without certificate
generation. In this regard, methods for certifying universal-variable
expansion should be developed~\cite{BubeckBuningSAT07} and other
techniques, such as hyper-binary resolution, must be certified.

Last but not least, methods for solving QBF were generalized to
domains such as SMT or
verification~\cite{cheng2013,DBLP:conf/ifm/MorgensternGS13}.  We may
expect that the contributions made by this paper will also be helpful
for these works.

  \section*{Acknowledgments}
  We thank Armin Biere and Allen Van Gelder for helpful conversations on QBF
  and preprocessing.
  This work is partially supported by SFI PI grant BEACON (09/IN.1/I2618),
  FCT grants ATTEST (CMU-PT/ELE/0009/2009), POLARIS (PTDC/EIA-CCO/123051/2010),
  and INESC-ID's multiannual PIDDAC funding PEst-OE/EEI/LA0021/2011.
  Grigore was supported by the EPSRC Programme Grant `Resource Reasoning' (EP/H008373/2).
\bibliographystyle{splncs03}
\bibliography{refs}
\end{document}